\documentclass[journal]{IEEEtran}
\IEEEoverridecommandlockouts
\usepackage{amsfonts,amsmath,amssymb} 

\newcommand{\dif}{\,\mathrm{d}}		  
\DeclareMathOperator{\diag}{diag}         

\usepackage{psfrag,color}
\usepackage{enumerate,cite,latexsym,graphicx}
\newtheorem{theorem}{Theorem}
\newtheorem{lemma}{Lemma}
\newtheorem{corollary}{Corollary}
\newtheorem{definition}{Definition}

\newtheorem{remark}{Remark}

\title{\LARGE \bf 
Quantum Noises, Physical Realizability and Coherent Quantum Feedback Control
}

\author{Shanon L.~Vuglar, Ian R.~Petersen%
\thanks{Shanon L. Vuglar is with the School of Engineering and Information 
Technology, 
        University of New South Wales at the Australian Defence Force Academy, Canberra ACT 2600, Australia.
         {\tt\small shanonvuglar@vuglar.com} }%
\thanks{Ian R. Petersen is with the School of Information Technology and Electrical Engineering, 
        University of New South Wales at the Australian Defence Force Academy, Canberra ACT 2600, Australia.
         {\tt\small i.r.petersen@gmail.com} } %
	 \thanks{The authors gratefully acknowledge support by  the Australian Research Council and the Air Force Office of Scientific Research (Grant Nos. FA2386-09-1-4089 and FA2386-12-1-4075).}}%

\begin{document}

\maketitle
\thispagestyle{empty}
\pagestyle{empty}

\begin{abstract}
Physical Realizability addresses the question of whether it is possible to 
implement a given linear time invariant (LTI) system as a quantum system. 
A given synthesized quantum controller described by a set of stochastic 
differential equations does not necessarily correspond to 
a physically meaningful quantum 
system. However, if additional quantum noises are permitted in the 
implementation, it is always possible to implement an arbitrary LTI system as a 
quantum system. In this paper, we give an expression for the number of 
introduced noise channels required to implement a given LTI system as a quantum system. 
We then consider the special case where only the transfer function 
to be implemented is of interest. 
We give results showing when it is possible to implement a 
transfer function as a quantum system by introducing the same number of 
quantum noises as there are system outputs.
Finally, we demonstrate the utility of these results by providing an algorithm
for obtaining a suboptimal solution to a coherent quantum LQG control problem.
\end{abstract}

\section{Introduction} \label{sec:intro}
For systems where it is necessary to 
consider quantum effects, 
the laws of quantum mechanics introduce new considerations not present in 
classical controller synthesis problems. 
The presence of quantum noises~\cite{GZ00} introduce fundamental limits on controller performance. 
Furthermore, the requirement for unitary evolution, the non-commuting nature of quantum observables, 
and the requirement for commutation relations to be preserved as systems evolve (see for example 
\cite{AFP09})
lead to the notion of \emph{physical realizability} 
 \cite{JNP1,ShP12}.  
This is the property that a given 
system model represents the dynamics of a physically meaningful quantum system.
Controller synthesis and optimization 
problems that are well understood in the classical regime 
can become difficult when restricting their solutions to physically realizable quantum 
controllers. 
New and tractable methods are required for these quantum controller synthesis problems.

It is useful to draw a distinction between measurement based quantum feedback control 
and coherent quantum feedback control.
In measurement based quantum feedback control, measurements of observables of a quantum system are used to 
apply feedback via a classical controller. While the closed loop system is modeled and analyzed 
in a quantum setting, the results of the measurements are classical 
signals and the controller can be implemented using analog or digital electronics.
This paper addresses the alternative coherent quantum feedback approach illustrated in 
Figure \ref{fig:coherent} in which quantum systems are 
interconnected directly, avoiding measurement; e.g., see \cite{NJP1}.

\begin{figure}[h]
	\centering
	\includegraphics[trim = 0mm 0mm 0mm -10mm, scale=0.4]{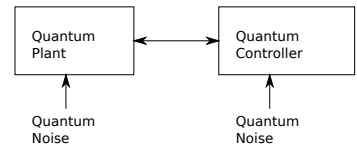}
	\caption{A Coherent Quantum Control Scheme \label{fig:coherent} 
}
\end{figure}

Both forms of quantum feedback control are relevant to a diverse range of applications which 
take advantage of quantum effects. These applications include  
quantum computing, quantum communications, quantum cryptography and 
precision metrology such as gravity wave detection.
Measurement based quantum feedback 
control is well understood (see for example~\cite{WM10}) and has been used successfully to 
manipulate quantum effects. For example, in \cite{VMS12} a qubit was maintained in an oscillating superposition 
state. This important result is relevant to the field of quantum computing. 

The majority of experimental results to date 
have focused on measurement based feedback control. 
%
%
%
%
Coherent quantum feedback control presents additional challenges in that the controller 
must be physically realizable. However, coherent quantum feedback control may offer several 
advantages over its measurement based counterpart. 
Firstly, coherent feedback avoids the collapse of 
the quantum state and
the loss of quantum information associated with the use of quantum measurement. 
This is particularly relevant to quantum 
computing where quantum states need to be maintained and manipulated. 
Secondly, implementing 
coherent controllers may introduce fewer quantum noise channels than the measurement process 
and this in turn may lead to better control system performance.
Finally, it may be that there are technical or experimental benefits in 
implementing a controller as a quantum system.  
For example, the use of a coherent controller may result in advantages in terms of the speed 
of control. Also, the experimental setup may make measurement impractical.

Coherent quantum feedback control 
is generating increasing interest within the research community 
(\!\cite{JNP1,MaP4,NJP1,MAB08}) and central to this area is the notion of physical realizability.
In the classical setting, we regard controllers as always being 
possible to implement. 
In the quantum setting, a 
given synthesized quantum controller described by a set of stochastic 
differential equations cannot always be implemented by a physically meaningful quantum 
system. Several recent papers have addressed this issue of physical 
realizability \cite{JNP1,ShP12,MaP3,VuP12b}, giving conditions for 
when a given system is physically realizable. 
Other papers \cite{NJD09,NUR10,NUR10A,Pet11} have given algorithms 
for experimentally implementing several classes of physically realizable quantum systems.

An importance difference between classical and quantum controller synthesis is that in the case 
of coherent quantum feedback control, implementing a controller as a quantum system 
may require the introduction of quantum vacuum noises. 
To see how this might arise, consider the following  
example from quantum optics.

\begin{figure}[h]
	\centering
	\includegraphics[trim = 0mm 0mm 0mm -10mm, scale=0.4]{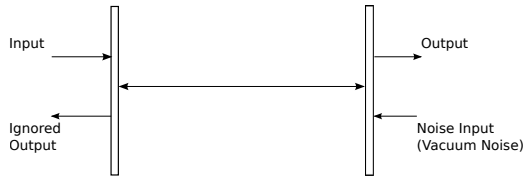}
	\caption{\label{fig} An Optical Cavity}
\end{figure}

Suppose that as part of the controller implementation process, 
the design calls for a laser beam to be 
passed through an optical cavity as shown in Figure \ref{fig}. 
Here a naive approach would 
be to consider this device as having a single input and single output. 
However, the laws of quantum mechanics imply that 
there is a second input to the cavity. Indeed, the mirror on the right, which 
produces the output, also causes the cavity to be coupled to a vacuum 
noise input; e.g. \cite{GZ00}. To obtain correct results when 
modeling such a system, it is essential to take this additional vacuum 
noise source into account.

Utilizing well established controller synthesis methods (such as $H^\infty$ controller synthesis) 
and modifying the 
classical solutions by incorporating additional quantum noises to obtain 
physically realizable quantum systems provides a tractable approach to coherent quantum 
controller design; e.g., see \cite{JNP1}. This approach requires a method for 
determining 
how many additional quantum noises are necessary for physical realizability, 
and for constructing the resulting quantum systems.

In \cite{JNP1}, the authors demonstrated that it is 
always possible to implement an arbitrary, strictly proper, linear time invariant~(LTI) system as a quantum system 
by introducing a sufficient number of quantum vacuum noise channels. 
It is straightforward to obtain upper and lower bounds on the number of 
introduced quantum noises that are necessary to obtain physical realizability. Since these noises 
place limits on the achievable controller performance, it is desirable to minimize the number 
of these introduced noises. This paper extends the result in~\cite{JNP1} 
to determine the number of introduced quantum noises that are necessary 
to implement a given, strictly proper, LTI system. 
Also, our result extends the construction method in \cite{JNP1} to give a
construction that only introduces as many  
quantum vacuum noises as are necessary to make that system physically realizable. 

We also consider the 
special case in which we are only interested in physically realizing a transfer function, as opposed 
to a specific state space realization.
Note that the number of introduced quantum noises necessary to 
physically realize a strictly proper LTI system 
must be at least as many as the output dimension.
We provide a condition, 
under which a strictly proper transfer function can be physically realized 
with the number of introduced quantum noises being equal to the output dimension. 
This condition is given in terms of a non-standard algebraic Riccati equation. 
We then provide conditions for the existence of a suitable solution 
to this Riccati equation. This leads to a numerical solution to the question of 
whether a particular strictly proper transfer function is physically realizable with the number of introduced 
quantum noises being equal to the output dimension. 

Preliminary conference versions of the results of this paper have appeared in 
\cite{VuP11a,VuP12a,VuP12c}. Here, we provide 
detailed proofs not included in those conference papers. Furthermore, we demonstrate the utility of our 
main results by providing an algorithm to 
obtain a suboptimal solution to a coherent quantum linear quadratic Gaussian (LQG) problem. 
This algorithm, and the example demonstrating its application, 
did not appear in the conference papers.

The remainder of the paper proceeds as follows. In Section~\ref{sec:model}, 
we describe the quantum systems considered throughout this paper. In Section
\ref {sec:realizability}, we recall the definition of physically realizable 
systems and outline relevant previous results on this topic. 
In Section \ref{sec:result}, we present our 
main results. We first consider the problem of implementing a 
particular state space model as a quantum system and the number of  
introduced quantum noises necessary to do so. We then consider the special 
case where a transfer function is to be physically realized. 
We give results regarding when such a transfer function 
is physically realizable with the number of introduced noises being equal to the  
dimension of the system output. 
In Section \ref{sec:lqg}, we demonstrate the utility of our results by presenting 
an algorithm for finding a quantum controller which is a suboptimal solution to 
a coherent quantum LQG problem.
An example demonstrating our algorithm followed by our conclusion are then given 
in Sections \ref{sec:ex} and \ref{sec:conc} respectively.

\section{Quantum Systems} \label{sec:model}
\subsection{General Quantum System Model} \label{subsec:generalmodel}
Open quantum harmonic oscillators represent an 
important class of quantum systems. Such systems can be described by 
quantum stochastic differential equations (QSDEs) of the following form (see ~\cite{JNP1}):
\begin{eqnarray}
	\dif x(t) &=& A x(t) \dif t + B \dif w(t) 
		; \nonumber \\
	\dif y(t) &=& C x(t) \dif t + D \dif w(t). 
	\label{eqn:model}
\end{eqnarray}

Here, $x(t) = \begin{bmatrix}x_1(t) & \cdots & x_n(t)\end{bmatrix}^T$
is a column vector of $n$ self-adjoint system variables which are operators on an underlying 
Hilbert space. Being quantum in nature, these 
variables do not commute in general. 
The commutation relations for these variables are described by 
a real skew-symmetric matrix~$\Theta$:
\begin{equation*}
		\begin{bmatrix}x_i(t),x_j(t)\end{bmatrix} 
		= x_i(t)x_j(t) - x_j(t)x_i(t)
		= 2i\Theta_{ij}.
\end{equation*}

Similarly, 
$\dif w(t)$ is a column vector of $n_w$ self-adjoint, non-commutative operators 
representing the input to the system and  
$\dif y(t)$ is a column vector of $n_y$ self-adjoint, non-commutative operators 
representing the output of the system. Their commutation relations are as 
follows:
$$\begin{bmatrix}\dif w_i(t),\dif w_j(t)\end{bmatrix} = 2i\Theta_{w,ij} \dif t;$$
$$\begin{bmatrix}\dif y_i(t),\dif y_j(t)\end{bmatrix} = 2i\Theta_{y,ij} \dif t;$$
where $\Theta_w$ and $\Theta_y$ are real skew symmetric matrices.

The input signals $\dif w(t)$ are assumed to admit the decomposition
$$\dif w(t) = \beta_w(t) \dif t + \dif \tilde{w}(t)$$ 
where the self-adjoint, adapted process $\beta_w(t)$ is  
the signal part of $\dif w(t)$ and $\dif \tilde{w}(t)$
is the noise part of $\dif w(t)$. Here, $\beta_w(t)$ is 
assumed to commute with $x(t)$. The vector $\dif \tilde{w}(t)$ is a quantum Wiener process 
with Ito products
$$\dif \tilde{w}(t) \dif \tilde{w}^T(t) = F_{\tilde{w}} \dif t$$
where $F_{\tilde{w}}$ is a non-negative Hermitian matrix. Let 
$F_{\tilde{w}} = S_{\tilde{w}} + T_{\tilde{w}}$, where $S_{\tilde{w}}$ is real 
and $T_{\tilde{w}}$ is imaginary. Then 
$S_{\tilde{w}}$ describes the intensity of the quantum Wiener process and is the 
quantum analog of the intensity matrix for a classical Wiener process. 
The commutation relations for $\dif \tilde{w}$ are determined by $T_{\tilde{w}}$:
$$\left[ \dif \tilde{w}(t), \dif \tilde{w}^T(t) \right] = 2 T_{\tilde{w}} \dif t.$$
Since $\beta_w(t)$ is an adapted process, $\beta_w(t)$ commutes with $\dif \tilde{w}(t)$ 
for all $ t \ge 0$. Also, $\dif \tilde{w}(t)$ commutes with $x(t)$. 

Finally, $n,n_w$ and $n_y$ are even (this is because in the quantum harmonic oscillator, the system 
variables always occur as conjugate pairs, see \cite{WM10}) 
and $A,B,C$ and $D$ are appropriately 
dimensioned real matrices describing the dynamics of the system. For further
details regarding these models, see~\cite{JNP1}. 

\begin{remark}
	While it is always possible to describe a collection of quantum harmonic oscillators 
	by QSDEs of the form~(\ref{eqn:model}), not all QSDEs 
	of this form 
	correspond to a collection of quantum harmonic oscillators. The property that the 
	QSDEs~(\ref{eqn:model}) correspond to a collection of quantum harmonic oscillators is called 
	physical realizability and is addressed in greater detail in 
	Section~\ref{sec:realizability}.
\end{remark}

We will now further restrict our attention within the class of quantum systems described 
above.

\subsection{A Class of Quantum System Models} \label{subsec:specialmodel}
This paper addresses the problem of implementing an arbitrary, strictly proper, 
LTI system as a quantum system 
(for example when implementing a coherent controller) by introducing 
vacuum noise sources. The resulting quantum 
systems are described by the following QSDEs 
which are a special case of (\ref{eqn:model}):
\begin{eqnarray}
	\dif x(t) &=& A x(t) \dif t + B_u \dif u(t) \nonumber \\
	&& {} + B_{v_1} \dif v_1(t) +  B_{v_2} \dif v_2(t)  
		; \nonumber \\
	\dif y(t) &=& C x(t) \dif t + \dif v_1(t).
	\label{eqn:model2}
\end{eqnarray}
Here, $\dif u(t)$ (a column vector with $n_u$ components) 
represents the inputs to the system and, like $\dif w(t)$ in (\ref{eqn:model}), 
is assumed to admit the 
decomposition $\dif u(t) = \beta_u(t) \dif t + \dif \tilde{u}(t).$
Also, $\dif v_1(t)$ and $\dif v_2(t)$ (column vectors with $n_{v_1}$ and $n_{v_2}$ components respectively) 
are quantum Wiener processes 
corresponding to the introduced vacuum noise inputs. For convenience, the vacuum noises are 
partitioned into two vectors $\dif v_1(t)$ and $\dif v_2(t)$ such that $n_{v_1} = n_u$. 
Then, $n_v = n_{v_1} + n_{v_2}$ is the total number of introduced vacuum noise inputs. 
Subsequently, we will refer to $\dif v_1$ as the \emph{direct feedthrough quantum noises} and to 
$\dif v_2$ as the \emph{additional quantum noises}. 
Also, 
$F_{\tilde{u}}$, $S_{\tilde{u}}$, $T_{\tilde{u}}$, 
$F_{v_1}$, $S_{v_1}$, $T_{v_1}$, $F_{v_2}$, $S_{v_2}$, and $T_{v_2}$ are defined for $\dif \tilde{u}(t)$, 
$\dif v_1(t)$ and $\dif v_2(t)$ respectively  
as $F_{\tilde{w}}$, $S_{\tilde{w}}$ and $T_{\tilde{w}}$ were for $\dif \tilde{w}(t)$ 
in~(\ref{eqn:model}). 
Furthermore, we assume that $F_{v1}$ and 
$F_{v2}$ are appropriately dimensioned block diagonal matrices with each diagonal 
block equal to 
$ \left[ \begin{smallmatrix}1 & i \\ -i & 1
	\end{smallmatrix} \right]$. 
This assumption corresponds to the fact that $\dif v_1$ and $\dif v_2$ 
represent vacuum noises \cite{JNP1}.
The remaining symbols have the same meanings as in (\ref{eqn:model}). We restrict our 
attention to the case where $n_y = n_u$.

\section{Physical Realizability} \label{sec:realizability}
\subsection{Definitions} \label{subsec:def}
As in \cite{JNP1,MaP3,MaP4,VuP11a,VuP12a,VuP12c}, 
the concept of physical realizability
means that the system dynamics described by the QSDEs~(\ref{eqn:model}) 
correspond to those of a collection of open quantum harmonic oscillators. 
Here, we slightly modify the definition of physically realizable given in 
\cite [Definition 3.1]{JNP1}. In \cite{JNP1} 
both fully quantum systems and hybrid systems with quantum and classical 
degrees of freedom are considered. However, 
we restrict our definition of physically realizable 
systems to those that are fully quantum.

\begin{definition}
	The system variables $x$ are said to satisfy the canonical commutation relations 
	if $$\begin{bmatrix}x_i(t),x_j(t)\end{bmatrix} = 2i\Theta_{ij}$$
	where $\Theta$ is of the form 
	\begin{equation}
	\Theta = \begin{bmatrix} 
	J & 0 & \cdots & 0 \\
	0 & J & \cdots & 0 \\
	\vdots  & \vdots & \ddots & \vdots \\
	0 & 0 & \cdots & J
	\end{bmatrix};
	\qquad 
	J = \begin{bmatrix}0 & 1 \\ -1 & 0
	\end{bmatrix}.   
		\label{eqn:ctheta}
	\end{equation}	
	This corresponds to the case where $x$ consists of pairs of position and 
	momentum operators: $\begin{bmatrix} q_1(t) & p_1(t) & q_2(t) & p_2(t) & \dots \end{bmatrix}^T$.
\end{definition}

\begin{definition} \label{def:pr}
	The system described by (\ref{eqn:model}) is \emph{physically realizable} if 
	$\Theta$ is of the form (\ref{eqn:ctheta}) and
	there exists 
	a quadratic Hamiltonian operator $\mathcal{H} = \frac{1}{2} x(0)^TRx(0)$,
	where $R$ is a real, symmetric, $n \times n$ matrix, and  
	a coupling operator vector $\mathcal{L} = \Lambda x(0)$, 
	where $\Lambda$ is a complex-valued $\frac{1}{2} n_w \times n$ 
	coupling matrix  
	such that the matrices $A$, $B$, $C$ and $D$ are given by:
	\begin{subequations}
	\begin{align}
		A &= 2 \Theta \left(R + \mathfrak{Im}\left(
			\Lambda^{\dagger}\Lambda \right) \right);
			\label{eqn:a} 
			\\
		B &= 2i \Theta \begin{bmatrix}
			-\Lambda^{\dagger} & \Lambda^T \end{bmatrix}\Gamma; 
			\label{eqn:b} 
			\\
		C &= P^T \begin{bmatrix}
			\Sigma_{n_y} & 0 \\ 0 & \Sigma_{n_y} \end{bmatrix}
			\begin{bmatrix} \Lambda + \Lambda^\# \\
				-i\Lambda + i\Lambda^\# \end{bmatrix};
			\label{eqn:c} 
			\\
		D &= \begin{bmatrix} I_{n_y \times n_y} &
			0_{n_y \times \left(n_w - n_y\right)} 
			\end{bmatrix}.
			\label{eqn:d} 
	\end{align}
	\end{subequations}
	Here: $\Gamma_{n_w \times n_w} = P \diag (M)$; 
		$M = \frac{1}{2}\left[ \begin{smallmatrix}1 & i \\ 1 & -i 
			\end{smallmatrix} \right]$;  
		$\Sigma_{n_y} = \begin{bmatrix}
			I_{\frac{1}{2}n_y \times \frac{1}{2}n_y} &
			0_{\frac{1}{2}n_y \times \frac{1}{2}\left(
				n_w - n_y \right) } \end{bmatrix}$;  
	$P$ is the appropriately dimensioned square permutation 
	matrix such that 
	$ P \begin{bmatrix}a_1 & a_2 & \cdots & a_{2m} \end{bmatrix} $ \linebreak
$	=\begin{bmatrix}a_1 & a_3 & \cdots & a_{2m-1} 
	a_2 & a_4 & \cdots & a_{2m} \end{bmatrix}$
	and 
	$\diag (M)$ is an appropriately dimensioned square block diagonal 
	matrix with each diagonal block equal to the matrix $M$. (Note that the  
	dimensions of $P$ and $\diag (M)$ can always be determined from the
	context in which they appear.) $\mathfrak{Im}\left(.\right)$ 
	denotes the imaginary part of a matrix and ${}^\dagger$ denotes the 
	complex conjugate transpose of a matrix.
\end{definition}

\begin{remark}
	This definition amounts to saying that a system~(\ref{eqn:model}) is physically realizable 
	if and only if it corresponds to a collection of open quantum harmonic oscillators \cite{JNP1}.
\end{remark}

We now apply this definition to the class of quantum systems (\ref{eqn:model2}). 
	The system (\ref{eqn:model2}) is physically realizable if 
	there exists 
	a real, symmetric, $n \times n$ matrix $R$, and  
	a complex-valued $\frac{1}{2}\left( n_{v_1} + n_{v_2} + n_u\right)	\times n$ 
	 matrix $\Lambda$ such 
	that the matrices $A$, $B_u$, $B_{v_1}$, $B_{v_2}$ and $C$ are given by: 
	\begin{subequations}
	\begin{align}
		A &= 2 \Theta \left(R + \mathfrak{Im}\left(
			\Lambda^{\dagger}\Lambda \right) \right);
			\label{eqn:a2} \\
		\begin{bmatrix}
			B_{v_1} & B_{v_2} & B_u 
		\end{bmatrix}			
			&= 2i \Theta \begin{bmatrix}
			-\Lambda^{\dagger} & \Lambda^T \end{bmatrix}\Gamma; 
			\label{eqn:b2} \\
		C &= P^T \begin{bmatrix}
			\Sigma_{n_y} & 0 \\ 0 & \Sigma_{n_y} \end{bmatrix}
			\begin{bmatrix} \Lambda + \Lambda^\# \\
				-i\Lambda + i\Lambda^\# \end{bmatrix},
			\label{eqn:c2}
	\end{align}
	\end{subequations}
	where $\Theta$ is of the form (\ref{eqn:ctheta}). 
	Here, 
		$\Sigma_{n_y} = \begin{bmatrix}
			I_{\frac{1}{2}n_y \times \frac{1}{2}n_y} &
			0_{\frac{1}{2}n_y \times \frac{1}{2}\left(
			n_{v_1} + n_{v_2} + n_u	- n_y \right) } \end{bmatrix}$.

\begin{theorem}
	(see \cite [Theorem 3.4]{JNP1})
	The system (\ref{eqn:model}) is physically realizable if and only
	if:
	\begin{eqnarray*}
		0 &=& iA\Theta + i\Theta A^T + BT_{\tilde{w}}B^T; \\
		B \begin{bmatrix} I \\ 0 \end{bmatrix} 
		&=& \Theta C^T \mbox{diag}(J);
	\end{eqnarray*}
	where $\Theta$ is defined as in (\ref{eqn:ctheta}) and $D$ satisfies (\ref{eqn:d}).
\end{theorem}

\begin{corollary}
	The system (\ref{eqn:model2}) is physically realizable if and only
	if:
	\begin{eqnarray*}
		0 &=& iA\Theta + i\Theta A^T \nonumber 
		+ B_{v_1}T_{v_1}B_{v_1}^T \\
		&& {} 
		+ B_{v_2}T_{v_2}B_{v_2}^T 
		+ B_uT_{\tilde{u}}B_u^T; \\
		B_{v_1} &=& \Theta C^T \mbox{diag}(J),
	\end{eqnarray*}
	where $\Theta$ is defined as in (\ref{eqn:ctheta}). 
\end{corollary}


\subsection{Previous Results} \label{subsec:previous}
In \cite{JNP1}, it was demonstrated that by introducing a sufficient number of 
vacuum noises, an arbitrary LTI system could be made 
physically realizable. In particular, the following lemma 
relating to the physical realizability of a purely quantum controller was proved. 

\begin{lemma}
	(See \cite [Lemma 5.6]{JNP1}).
	Let $F_{\tilde{u}}$ be a block diagonal matrix with each diagonal block 
	equal to $I + iJ$, and let 
	$A,B$ and $C$ be matrices such that
	$A \in \mathbb{R}^{n\times n}, 
	B \in \mathbb{R}^{n\times n_u}, 
	C \in \mathbb{R}^{n_y\times n}.$
	Also, let $\Theta$ be defined as in (\ref{eqn:ctheta}). Then there exists an even 
	integer	$n_{v_2} \ge 0$ and matrices $B_{v_1} \in \mathbb{R}^{n\times n_{v_1}}$,
	$B_{v_2} \in \mathbb{R}^{n\times n_{v_2}}$,
	such that the system~(\ref{eqn:model2}) is physically realizable.
\end{lemma}

\begin{remark} \label{rem:minimal}
It follows from \cite [Theorem 3.4]{JNP1}, that 
for a system described by a strictly proper 
transfer function 
$$G(s) = C \left( sI - A \right)^{-1} B_u,$$
the dimension of the system output $n_y$ is a 
lower bound on the total number of introduced vacuum noises $n_v$ that are necessary for 
the system to be physically realizable. That is, the direct feedthrough quantum 
noises $\dif v_1$ in the system~(\ref{eqn:model2}) are necessary, 
but may not be sufficient for physical realizability. 
We are also interested in the situation in which the presence of the noises  $\dif v_1$ is 
sufficient for physical realizability and the noises $\dif v_2$ are not needed. 
In this case, we say that the LTI system is physically realizable with 
no additional vacuum noises. Physically realizing a system with 
\emph{minimal additional noises} means to 
implement the system as a quantum system (\ref{eqn:model2}) by only introducing the number 
of additional vacuum noises $n_{v_2} \ge 0$, that are necessary 
for physical realizability. 
\end{remark}

\section{Main Result} \label{sec:result}
\subsection{General Case - Implementing a State Space representation} \label{subsec:gc}
In this section, we give a method to physically realize a strictly proper LTI system 
\begin{eqnarray*}
	\dif x &=&  A x \dif t + B_u \dif u \\ 
	\dif y &=&  C x \dif t
\end{eqnarray*}
with minimal additional quantum noises. 
The remainder of this 
section is structured as follows. We first give our algorithm. 
We then formally state our result as a theorem. The subsequent proof 
of the theorem justifies our algorithm.

The algorithm for obtaining a physically realizable system~(\ref{eqn:model2}) with 
minimal additional quantum noises proceeds as follows:

\begin{enumerate}
\item 
Construct the matrix
\begin{equation} \tilde{S} = \Theta B_u \Theta_{n_u} B_u^T \Theta - \Theta A - A^T \Theta
	- C^T \Theta_{n_y} C. \label{eqn:s} \end{equation}
Here $\Theta_{n_u}$ and $\Theta_{n_y}$ are 
commutation matrices of the form~(\ref{eqn:ctheta}) of dimensions $n_u \times n_u$ and $n_y \times n_y$ 
respectively. 
\item 
	Find the rank of the matrix $\tilde{S}$: $r = \mbox{rank} \left[ \tilde{S} \right]$.
	Now $n_{v_2} = r$. That is, $n_u$ direct feedthrough quantum noises, 
	and $r$ additional quantum noises, are necessary 
	for the existence of $B_{v_1}$, $B_{v_2}$ such that
	the system (\ref{eqn:model2}) is  
	physically realizable.
	This gives $n_v = n_u + r$.
\item
	Calculate $S = \frac{i}{4}\tilde{S}$.
\item
	Construct the singular value decomposition (SVD) for $S$: 
	$S = U^{\dagger}DU$. Here $D$ is diagonal and $U$ is unitary.
\item
	Construct $\Lambda_{b1} = \left( \left| D \right| + D \right)^{\frac{1}{2}} U$ 
	where $\left| D \right|$ is the diagonal matrix with entries equal to 
	the absolute values of the corresponding entries in $D$.
\item
	Construct $B_{v_1}$ and $B_{v_2}$ as follows: 
	\begin{eqnarray*}
		B_{v_1} &=& \Theta C^T \diag (J); \\
		B_{v_2} &=& 2i\Theta \begin{bmatrix}
			-\Lambda^{\dagger}_{b1} & \Lambda^T_{b1} \end{bmatrix}
			P \nonumber \diag (M).
	\end{eqnarray*}
\end{enumerate}
The system~(\ref{eqn:model2}) with $\left\{ A,B_u,C \right\}$ given and 
with $B_{v_1}$, $B_{v_2}$ so constructed is physically realizable. 
We now give a theorem which formally states that $n_v = n_u + r$ introduced noises are 
necessary and sufficient for physical realizability. 
The construction of $B_{v_1}$ and $B_{v_2}$ above follows from the 
proof of the theorem.

\begin{theorem}
	\label{thm:main}
	Consider a strictly proper LTI system defined by given matrices  
	$A,B_u$ and $C$. 
	There exist matrices $B_{v_1}$ and $B_{v_2}$ such that the corresponding 
	system (\ref{eqn:model2}) is physically 
	realizable and with 
	$n_{v_2}$ equal to $r$ where $r$ 
	is the rank of the matrix 
	$\left( \Theta B_u \Theta_{n_u} B_u^T \Theta - \Theta A - A^T \Theta
	- C^T \Theta_{n_y} C \right)$. 
	Conversely, suppose that there exist matrices $B_{v_1}$ and $B_{v_2}$ such that 
	the corresponding system (\ref{eqn:model2}) is physically 
	realizable. 
	Then $n_{v_2} \ge r$. 
\end{theorem}

\begin{proof}
The proof is structured as follows. We first show that $n_u + r$ 
introduced quantum 
noises are sufficient for physical realizability. We then show that 
$n_u + r$ introduced quantum noises are necessary for physical realizability.
	
Following the method of \cite{JNP1}, the construction of the matrices $R$, $\Lambda$, $B_{v_1}$ and 
$B_{v_2}$ in (\ref{eqn:a2}) - (\ref{eqn:c2}) is as follows:  
	\begin{eqnarray*}
		R &=&  -\frac{1}{4}\left( \Theta A + A^T \Theta^T \right); \\
		\Lambda &=& \begin{bmatrix}
			\frac{1}{2}C^T P^T \begin{bmatrix}
				I \\ iI \end{bmatrix} &
			\Lambda^T_{b1} & \Lambda^T_{b2} \end{bmatrix}^T; \\
		B_{v_1} &=& \Theta C^T \diag (J); \\
		B_{v_2} &=& 2i\Theta \begin{bmatrix}
			-\Lambda^{\dagger}_{b1} & \Lambda^T_{b1} \end{bmatrix}
			P \nonumber \diag (M),
	\end{eqnarray*}
	where $\Theta$ is defined as in (\ref{eqn:ctheta}). 
	Here, 
	$$	\Lambda_{b2} = -i \begin{bmatrix} I &
			0 \end{bmatrix}
			P \diag (M) \nonumber B_u^T \Theta; $$
			and $\Lambda_{b1}$ is any complex $\frac{1}{2} n_{v_2}
	\times n$ matrix such that
	\begin{eqnarray}
		\Lambda_{b1}^{\dagger}\Lambda_{b1} &=&  \Xi_1 \nonumber \\ 
		&&{} + i \bigg( 
		\frac{A^T \Theta^T  - \Theta A}{4} - \frac{1}{4}C^T P^T
		\begin{bmatrix}0 & I \\ -I & 0 \end{bmatrix} PC \nonumber \\ 
		&&{} - \mathfrak{Im}
		\left( \Lambda^{\dagger}_{b2} \Lambda_{b2} \right) 
		\bigg) \label{eqn:XX}
	\end{eqnarray}
	where $\Xi_1$ is any real symmetric $n \times n$ matrix such that
	$\Lambda_{b1}^{\dagger}\Lambda_{b1}$ is nonnegative definite.

The matrix $\Lambda_{b1}$ can be constructed as follows: first a real 
symmetric $n \times n$ matrix $\Xi_1$ is constructed such that
the right hand side of (\ref{eqn:XX}) is nonnegative definite.
Then $\Lambda_{b1}$ is constructed such that (\ref{eqn:XX}) holds. 

Note that $\Lambda_{b1}$ has $\frac{1}{2}n_{v_2}$ rows and thus
determines the number of additional quantum noises required in this 
implementation. 
\if {false}
In \cite{JNP1}, the construction required that 
$n_v \ge n_u + 2$ (i.e. $n_{v_2} \ge 2$), however in the special case that 
the imaginary terms on the right hand side of (\ref{eqn:XX}) sum to zero, 
$\Xi_1$ can be set to the zero matrix and we can allow $n_{v_2} = 0$ and hence 
the $B_{v_2}$ and $\dif v_2$ terms in (\ref{eqn:model2}) vanish.
\fi
We now provide a method for choosing $\Xi_1$ and 
$\Lambda_{b1}$ to obtain the result.

It is desired to construct $\Xi_1$ such that 
$\Lambda_{b1}^{\dagger}\Lambda_{b1}$ 
is of minimum rank. This will allow $\Lambda_{b1}$ to be constructed 
with the minimum number of rows. 
We make the following observations about the terms in equation (\ref{eqn:XX}):
\begin{eqnarray}
	-\Theta &=& \Theta^T; \label{eqn:p2} \\
	P^T \begin{bmatrix} 0 & I \\ -I & 0 \end{bmatrix} P 
	&=& \Theta_{n_y}; 
	\label{eqn:p3} \\
	\Lambda_{b2}^{\dagger} \Lambda_{b2} &=& 
	\Theta^{\dagger} B_u \Gamma^{\dagger}
	\begin{bmatrix}I \\
	0\end{bmatrix}
	\begin{bmatrix}I&
	0\end{bmatrix}
	\Gamma B_u^T \Theta \nonumber \\
	&=& 
	- \frac{1}{4}\Theta B_u 
	\diag \left( \begin{bmatrix} 1 & i \\ -i & 1 \end{bmatrix} \right)
	B_u^T \Theta; \nonumber \\
	\mathfrak{Im}\left(\Lambda^{\dagger}_{b2}\Lambda_{b2}\right) &=& 
	- \frac{1}{4}\Theta B_u 
	\diag \left( \begin{bmatrix} 0 & 1 \\ -1 & 0 \end{bmatrix} \right)
	B_u^T \Theta \nonumber \\
	&=& - \frac{1}{4}\Theta B_u \Theta_{n_u} B_u^T \Theta.
	\label{eqn:p4}
\end{eqnarray}
Substituting (\ref{eqn:p2}), (\ref{eqn:p3}) and (\ref{eqn:p4}) into
(\ref{eqn:XX}), we obtain
$$\Xi_2 =  \Xi_1 + \frac{i}{4} \tilde{S}$$
where 
$$\Xi_2 =  \Lambda_{b1}^{\dagger}\Lambda_{b1}$$ 
and $\tilde{S}$ is defined as in (\ref{eqn:s}).

Note that the matrix $\tilde{S}$
is real and skew symmetric. Thus
$ S = \frac{i}{4} \tilde{S} $
is Hermitian, has real eigenvalues and is diagonalizable: 
$S = U^{\dagger}DU$ where $D$ is diagonal and $U$ is unitary. 

We wish to find a real, symmetric matrix $\Xi_1$ such that 
$\Xi_2 = \Xi_1 + S$ is positive semi-definite and of minimum rank.
Let $\Xi_1 = U^\dagger \left| D \right| U$.
We claim that $\Xi_1$ is real and 
symmetric, that $\Xi_2 = \Xi_1 + S \ge 0$, and that $\Xi_2$ has rank equal to 
half that of $S$.

First, we show that this matrix $\Xi_1$ is real and symmetric. 
Observe that $\Xi_1 = {\Xi_1}^\dagger$ and $\Xi_1 \ge 0$.
Also: 
$${\Xi_1}^2 = U^\dagger {\left| D \right|}^2 U = U^\dagger D^2 U = S^2.$$
Here, $S$ is purely imaginary, thus $S^2$ is real and ${\Xi_1}^2 \ge 0$ is also 
real, and therefore has a real square root. From the 
uniqueness of the positive semi-definite square root of a positive 
semi-definite matrix 
\cite [Theorem 7.2.6] {HJ85} we conclude that $\Xi_1$ is real.
Further, since $\Xi_1$ is Hermitian, $\Xi_1$ is symmetric.

We now show that $\Xi_2$ has rank equal to 
half that of $S$ and
that $\Xi_2$ is positive semi-definite. We observe that 
$S$ is Hermitian and so its eigenvalues are real. Thus the 
eigenvalues of $\tilde{S}$ are purely imaginary. Also since $\tilde{S}$ 
is real, its eigenvalues 
occur in complex conjugate pairs. Thus $D$ is of the form:
$$D = \begin{bmatrix} 
	\lambda_1 & 0 & 0 & 0 & \cdots \\
	0 & -\lambda_1 & 0 & 0 & \cdots \\
	0 & 0 & \lambda_2 & 0 & \cdots \\
	0 & 0 & 0 & -\lambda_2 & \cdots \\
	\vdots & & & & \ddots
\end{bmatrix}; \lambda_i \ge 0.$$

$$\left| D \right| = \begin{bmatrix}
	\lambda_1 & 0 & 0 & 0 & \cdots \\
	0 & \lambda_1 & 0 & 0 & \cdots \\
	0 & 0 & \lambda_2 & 0 & \cdots \\
	0 & 0 & 0 & \lambda_2 & \cdots \\
	\vdots & & & & \ddots
\end{bmatrix}; \lambda_i \ge 0.$$

$$\left| D \right| + D = \begin{bmatrix}
	2\lambda_1  & 0 & 0 & 0 & \cdots \\
	0 & 0 & 0 & 0 & \cdots \\
	0 & 0 & 2\lambda_2 & 0 & \cdots \\
	0 & 0 & 0 & 0 & \cdots \\
	\vdots & & & & \ddots
\end{bmatrix}; \lambda_i \ge 0.$$

From this, it can be seen that $\left| D \right| + D$ has a rank which is half that of 
$D$. Since
\begin{eqnarray*}
	\Xi_2 &=& \Xi_1 + S \\
		&=& U^\dagger \left| D \right| U + U^\dagger D U \\
		&=& U^\dagger (\left| D \right| + D) U,
\end{eqnarray*}
it follows that $\Xi_2$ is positive semi-definite and has a rank which is half that of $S$. 

Since $S$ and $\tilde{S}$ 
have the same rank, $\Xi_2$ has rank $\frac{r}{2}$ where $r$ the 
rank of $\tilde{S}$. Since $\Xi_2 \ge 0$ has rank $\frac{r}{2}$, it is 
possible to construct $\Lambda_{b1}$ with $\frac{r}{2}$ rows, 
such that $\Xi_2 = \Lambda_{b1}^\dagger \Lambda_{b1}$. Recall that, 
$\Lambda_{b1}$ has $\frac{1}{2}n_{v_2}$ rows, and we 
have $n_{v_2} = r$. That is, 
the system is
physically 
realizable with the number of additional quantum noises $n_v$ 
equal to $n_u + r$ where $r$ 
is the rank of the matrix $\tilde{S}$ defined in (\ref{eqn:s}).

We now consider the second part of the theorem and show that $n_u + r$
introduced noises are necessary for physical realizability. 
To do so, it is sufficient to show that the 
number of columns of 
$\begin{bmatrix}B_{v_1} & B_{v_2} \end{bmatrix}$ must be greater than or 
	equal to $n_u + r$.

From (\ref{eqn:b2}), it can be shown that:
\begin{subequations}
\begin{eqnarray}
	B_u &=&  2i \Theta \begin{bmatrix} -\Lambda^\dagger_{b2} &
		\Lambda^T_{b2} \end{bmatrix} \Gamma; \\
	B_{{v_1}} &=&  2i \Theta \begin{bmatrix} -\Lambda^\dagger_{b0} &
		\Lambda^T_{b0} \end{bmatrix} \Gamma; \\
	B_{{v_2}} &=&  2i \Theta \begin{bmatrix} -\Lambda^\dagger_{b1} &
		\Lambda^T_{b1} \end{bmatrix} \Gamma; \label{eqn:b11}
\end{eqnarray}
\end{subequations}
where $$\Lambda = \begin{bmatrix} \Lambda_{b0} \\ 
	\Lambda_{b1} \\ \Lambda_{b2} \end{bmatrix}.$$ 
That is, $B_{v_2}$ has twice the number of columns as $\Lambda_{b1}$ has 
rows. Therefore, we wish to show 
that $\Lambda_{b1}$ has at least $\frac{r}{2}$ 
rows.

Consider,
$$
\mathfrak{Im}(\Lambda^\dagger \Lambda) =
\mathfrak{Im}(\Lambda^\dagger_{b0} \Lambda_{b0} ) +
\mathfrak{Im}(\Lambda^\dagger_{b1} \Lambda_{b1} ) +
\mathfrak{Im}(\Lambda^\dagger_{b2} \Lambda_{b2} ).$$
That is, 
\begin{equation}
\mathfrak{Im}(\Lambda^\dagger_{b1} \Lambda_{b1} ) =
\mathfrak{Im}(\Lambda^\dagger \Lambda) -
\mathfrak{Im}(\Lambda^\dagger_{b0} \Lambda_{b0} ) -
\mathfrak{Im}(\Lambda^\dagger_{b2} \Lambda_{b2} ).\label{eqn:p5}
\end{equation}
Rearranging (\ref{eqn:a}), we obtain 
$$\frac{1}{2}\Theta^{-1}A = R + \mathfrak{Im}(\Lambda^\dagger \Lambda),$$
where $R$ and $\mathfrak{Im}(\Lambda^\dagger \Lambda),$ are respectively the symmetric and 
skew-symmetric parts of the left hand side of this equation. 
From this, it can be shown that 
\begin{equation}
\mathfrak{Im}(\Lambda^\dagger \Lambda) = -\frac{1}{4}(\Theta A + A^T 
\Theta). \label{eqn:p6}
\end{equation}
Also using (\ref{eqn:c}) and (\ref{eqn:p3}), it is straightforward to verify that
\begin{equation}
\mathfrak{Im}(\Lambda^\dagger_{b0} \Lambda_{b0}) = 
\frac{1}{4}C^T \Theta_{n_y} C. \label{eqn:p7}
\end{equation}
Substituting (\ref{eqn:p4}), (\ref{eqn:p6}) and (\ref{eqn:p7}) into (\ref{eqn:p5}) we obtain 
$$	\mathfrak{Im}(\Lambda^\dagger_{b1} \Lambda_{b1} ) 
= \frac{1}{4}\tilde{S}
$$
where $\tilde{S}$ is defined as in (\ref{eqn:s}).
That is, 
$$\Lambda^\dagger_{b1} \Lambda_{b1} = \Xi_1 + \frac{i}{4}\tilde{S},$$
where $\Xi_1$ is the real part of $\Lambda^\dagger_{b1} \Lambda_{b1}$. 

Now using \cite [Fact 2.17.3] {BER05}, we observe that 
\begin{eqnarray*}
	\mbox{rank} \left( \Lambda^\dagger_{b1} \Lambda_{b1} \right)
	&=& \mbox{rank} \left( \Xi_1 + i \frac {\tilde{S}}{4} \right) \\
	&=& \frac{1}{2} \mbox{rank} 
\begin{bmatrix} \Xi_1 & \frac {\tilde{S}}{4}  \\ 
	- \frac {\tilde{S}}{4} & \Xi_1 \end{bmatrix} \\
&\ge& \frac{1}{2} \mbox{rank} 
\begin{bmatrix} \Xi_1 & \frac {\tilde{S}}{4} \end{bmatrix} \\
&\ge& \frac{1}{2} \mbox{rank} \left[ \frac {\tilde{S}}{4} \right].
\end{eqnarray*}
That is, for any $\Xi_1$, 
$$\mbox{rank} \left(
\Lambda^\dagger_{b1} \Lambda_{b1} \right) \ge \frac{1}{2}\mbox{rank}[S].$$ 
This in turn implies that $\Lambda_{b1}$ has at least $\frac{r}{2}$ rows, 
where, $r$ is the rank of the matrix 
$\tilde{S}$ defined as in (\ref{eqn:s}). 
However, it follows from (\ref{eqn:b11}), that $B_{v_2}$ has twice as many 
columns as $\Lambda_{b1}$ has rows. 
That is, $B_{v_2}$ has at least $r$ columns and hence 
$\begin{bmatrix} B_{v_1} & B_{v_2} \end{bmatrix}$ has at least 
$n_u + r$ columns. Hence, the number of quantum noises $n_v$ is greater 
than or equal to $n_u + r$. This concludes the proof of the theorem.
\end{proof}

\subsection{Special Case - Physically Realizing a Transfer Function}

When designing LTI controllers, usually the transfer function of the controller
rather than its particular state space realization 
determines the closed loop performance. As such, the 
question of whether a particular transfer function is physically realizable 
may be of greater interest than whether a particular state space 
realization is physically realizable.

Therefore, we now turn our attention to the case in which we are interested in implementing an 
LTI quantum system with a specified strictly proper transfer function. This 
is equivalent to allowing a state transformation on the state space 
description of the system.
In particular, we consider the problem of whether a particular transfer 
function 
can be physically realized by only introducing the direct feedthrough 
quantum noises $\dif v_1$ in~(\ref{eqn:model2}). That is, 
without introducing any additional quantum noises $\dif v_2$. 

Here, we recall from Remark \ref{rem:minimal} that
for systems described by strictly proper 
transfer functions 
the direct feedthrough quantum noises are necessary for physical realizability.

Under some assumptions, it is possible 
to implement a specified transfer function as a physically realizable 
quantum system (\ref{eqn:model2}) where 
only the direct feedthrough quantum noises are introduced. 

%

\begin{theorem}
	\label{thm:tf}
	Consider a system with strictly proper transfer function matrix: 
	$$G(s) = \tilde{C}(sI - \tilde{A})^{-1}\tilde{B}_u.$$ 
	Suppose the algebraic Riccati equation (ARE) 
	\begin{equation}
		X \tilde{B}_u \Theta_{n_u} \tilde{B}_u^T X 
		- \tilde{A}^T X - X \tilde{A}	
		- \tilde{C}^T \Theta_{n_y} \tilde{C} = 0 
		\label{eqn:ric1}
	\end{equation}
	has a non-singular, real, skew-symmetric solution $X$. 
	Here, the matrices $\Theta_{n_u}$ and $\Theta_{n_y}$ are 
	defined as in (\ref{eqn:ctheta}). 
	Then there exists matrices $\left\{ A, B_u, C \right\}$ 
	such that 
	$$G(s) = C(sI - A)^{-1}B_u$$
	and the corresponding system (\ref{eqn:model2}) is physically realizable 
	with only the direct feedthrough quantum noises $\dif v_1$ and no 
	additional quantum noises $\dif v_2$. 
\end{theorem}


\begin{proof}
	First, note that for any $2m \times 2m$ non-singular real, 
	skew-symmetric matrix $X$ there exists a non-singular, real matrix $T$ 
	for which $X = T^T \Theta T$ where $\Theta$ is defined as in (\ref{eqn:ctheta}) 
	\cite[Corollary 8.24]{BAK02}. 
	Let, 
	\begin{eqnarray*}
		X &=&  T^T \Theta T; \qquad T \in \mathbb{R}^n; \quad \mbox{det } T \neq 0; \\
		A &=& T \tilde{A} T^{-1};\\
		B_u &=& T \tilde{B_u};\\
		C &=& \tilde{C} T^{-1};\\
		B_{v_1} &=& \Theta C^T \diag (J).
	\end{eqnarray*}

	The result now follows by applying Theorem \ref{thm:main}. 
	Indeed, (\ref{eqn:ric1}) implies: 
	\begin{eqnarray*}
		0 
		&=&  X \tilde{B}_u \Theta_{n_u} {\tilde{B}_u}^T X 
		- \tilde{A}^T X - X \tilde{A} - \tilde{C}^T \Theta_{n_y} \tilde{C} \\
		&=& T^T \Theta T\tilde{B}_u\Theta_{n_u} {\tilde{B}_u}^T T^T\Theta T - \tilde{A}^T T^T\Theta T 
		-T^T \Theta T \tilde{A} \\ 
		&& - \tilde{C}^T \Theta_{n_y} \tilde{C} \\
		&=& \Theta T\tilde{B}_u\Theta_{n_u} {\tilde{B}_u}^T T^T\Theta  
		- \left( T^T \right)^{-1} \tilde{A}^T T^T\Theta  
		- \Theta T \tilde{A} T^{-1} \\ 
		&& - \left( T^T \right) ^{-1} \tilde{C}^T \Theta_{n_y} \tilde{C} T^{-1}\\
		&=& \Theta T\tilde{B}_u\Theta_{n_u} \left( T\tilde{B}_u \right) ^T \Theta  
		- \left( T\tilde{A}T^{-1} \right)^T \Theta  
		- \Theta T \tilde{A} T^{-1} \\ 
		&& - \left( \tilde{C}T^{-1} \right) ^T \Theta_{n_y} \tilde{C} T^{-1}\\
		&=& \Theta B_u \Theta_{n_u} {B_u}^T \Theta  
		- A^T \Theta  
		- \Theta A  
		- C^T \Theta_{n_y} C.\\
	\end{eqnarray*}
	That is, the matrix $\tilde{S}$ defined in (\ref{eqn:s}) has rank zero. Applying
	Theorem \ref{thm:main}, we conclude that the system $\left\{ A,
	B_u, C \right\}$ can be physically realized with $n_{v_2} = 0$. 
\end{proof}

We now give conditions for when the ARE~
(\ref{eqn:ric1}) has a non-singular, real, skew symmetric solution $X$. 
The proof given below 
closely follows that in
\cite[Theorem 13.5]{ZDG96}.
This result also leads to a numerical procedure 
for physically realizing a strictly proper transfer function with the minimal number 
of additional quantum noises. 

For convenience, we define $\tilde{R} = -\tilde{B} \Theta_{n_u} \tilde{B}^T$, 
$\tilde{Q} = \tilde{C}^T \Theta_{n_y} \tilde{C}$ and 
rewrite (\ref{eqn:ric1}): 
\begin{equation}
	\label{eqn:ric2}
	\tilde{A}^T X + X \tilde{A} + X \tilde{R} X + \tilde{Q} = 0.
\end{equation}
Note that $\tilde{Q}$ and $\tilde{R}$ are skew symmetric.

Define 
\begin{equation}
	H = \begin{bmatrix}\tilde{A} & \tilde{R} \\ -\tilde{Q} & -\tilde{A}^T \end{bmatrix},
\label{eqn:h}
\end{equation} and $ 
	Z = -i \left[ \begin{smallmatrix} 0 & I \\ I & 0 \end{smallmatrix} 
		\right] $.
Note that  $Z^{-1} = Z^\dagger$, $(ZH)$ is skew symmetric, and 
$Z^{-1} H Z = Z^\dagger H Z = - H^\dagger.$
That is, $H$ and $-H^\dagger$ are similar, from which it follows that 
$\lambda$ is an eigenvalue of $H$ if and only if $-\lambda$ is. That is, the eigenvalues of 
$H$ are symmetric about the imaginary axis.

Assume $H$ has no eigenvalues on the imaginary axis and let $\chi_-(H)$ be the 
$n$-dimensional spectral subspace \cite{ZDG96} of H corresponding to its negative 
eigenvalues. We find a set of basis vectors for $\chi_-(H)$ and stack the basis vectors 
to form a matrix. Partitioning this matrix, we can write 
$\chi_-(H) = \mathrm{Im} \left[ \begin{smallmatrix}X_1 \\X_2\end{smallmatrix}
	\right]$ where $X_1,X_2 \in \mathbb{C}^{n\times n}$. 
	Here $\mathrm{Im} \left[ \begin{smallmatrix}X_1 \\X_2\end{smallmatrix}	\right]$ 
	denotes the subspace spanned by the columns of the matrix  
	$\left[ \begin{smallmatrix}X_1 \\X_2\end{smallmatrix}\right]$.

We assume $X_1$ is non-singular or equivalently that $\chi_-(H)$ and 
$\mathrm{Im} \left[ \begin{smallmatrix} 0 \\ I \end{smallmatrix} \right]$ 
are complementary subspaces. Then define $X = X_2 X_1^{-1}$. 
It follows that $X$ is uniquely determined by $H$. We will  
denote the corresponding function by $X = \mathrm{\emph{Ric}}(H)$ with the 
domain \emph{dom(Ric)} consisting of matrices $H$ of the form~(\ref{eqn:h}) satisfying the 
properties that $H$ has no purely imaginary eigenvalues, and that $X_1$ is 
non-singular.

\begin{theorem}
	\label{thm:tf2}
	Suppose $H \in \mathrm{\emph{dom(Ric)}}$ and $X = \mathrm{\emph{Ric}}
	(H)$. Then $X$ is skew-symmetric and solves the algebraic Riccati 
	equation 
	$\tilde{A}^T X + X \tilde{A} + X \tilde{R} X + \tilde{Q} = 0.$
\end{theorem}
\begin{proof}
	Let $X_1$, $X_2$ be as above. There exists a Hurwitz matrix $H_- \in 
	\mathbb{C}^{n\times n}$ such that 
	\begin{equation}
		\label{eqn:HH-}
	H \begin{bmatrix}X_1 \\X_2\end{bmatrix} = 
		\begin{bmatrix}X_1 \\X_2\end{bmatrix} H_-.
	\end{equation}
	Pre-multiply (\ref{eqn:HH-}) by 
$\left[ \begin{smallmatrix}X_1 \\X_2\end{smallmatrix} \right]^T Z$
to obtain
$$\left[ \begin{smallmatrix}X_1 \\X_2\end{smallmatrix} \right]^T Z
	H \left[ \begin{smallmatrix}X_1 \\X_2\end{smallmatrix} \right] = 
	\left[ \begin{smallmatrix}X_1 \\X_2\end{smallmatrix} \right]^T 
	Z\left[ \begin{smallmatrix}X_1 \\X_2\end{smallmatrix} \right] H_-.$$
Since $ZH$ is skew-symmetric, so are both sides of the above equation. From 
the right-hand side:
$$(X_2^T X_1 + X_1^T X_2) H_- = - H_-^\dagger (X_2^T X_1 + X_1^T X_2)^T.$$ 
This is a Lyapanov equation. Since $H_-$ is Hurwitz, the unique solution is 
$X_2^T X_1 + X_1^T X_2 = 0.$
That is, $X_1^T X_2$ is skew symmetric, and since $X_1$ is non-singular, 
$X = (X_1^{-1})^T (X_1^T X_2) X_1^{-1}$ is also skew-symmetric.

It remains to be shown that $X$ is a solution to (\ref{eqn:ric2}).
Post-multiplying (\ref{eqn:HH-}) by $X_1^{-1}$, we obtain 
	$$H \left[ \begin{smallmatrix}I \\X\end{smallmatrix} \right] = 
		\left[ \begin{smallmatrix}I \\X \end{smallmatrix} \right] 
		X_1 H_- X_1^{-1}$$
and pre-multiplying by $\left[ \begin{smallmatrix} X & -I \end{smallmatrix} 
	\right] $ gives 
	$$\left[ \begin{smallmatrix} X & -I \end{smallmatrix} \right] H 
	\left[ \begin{smallmatrix}I \\X \end{smallmatrix} \right] = 0,$$ 
which is precisely (\ref{eqn:ric2}).
\end{proof}
\begin{remark}
	The above proof also leads to a numerical procedure for solving the 
	ARE (\ref{eqn:ric1}) and hence solving the physical realizability 
	problem under consideration. This numerical procedure involves 
	solving the eigenvalue, eigenvector problem for the matrix H. 
	The following corollary, which follows directly from combining 
	Theorems \ref{thm:tf} and \ref{thm:tf2}, is the main result 
	of this subsection.
\end{remark}
\begin{corollary}
	Consider a system with strictly proper transfer function matrix: 
	$$G(s) = \tilde{C}(sI - \tilde{A})^{-1}\tilde{B}_u.$$
	Suppose $H \in \mathrm{\emph{dom(Ric)}}$ and
	$X = \mathrm{\emph{Ric}}(H)$ is non-singular where $H$ is defined as in 
	(\ref{eqn:h}). 
	Then there exists matrices $\left\{ A, B_u, C \right\}$ 
	such that 
	$$G(s) = C(sI - A)^{-1}B_u$$ 
	and the corresponding system (\ref{eqn:model2}) is physically realizable 
	with only the direct feedthrough quantum noises $\dif v_1$ and no 
	additional quantum noises $\dif v_2$. 
\end{corollary}

We now give the algorithm for solving the ARE (\ref{eqn:ric1}) and hence 
physically realizing a given transfer function by only introducing 
direct feedthrough quantum noises.
Suppose we wish to physically realize the transfer function 
$$G(s) = \tilde{C}(sI - \tilde{A})^{-1}\tilde{B}_u.$$
\begin{enumerate}
\item
	Construct the matrix 
	$H$ as in (\ref{eqn:h}). 
	Find the eigenvalues and eigenvectors of $H$.
	Check that $H$ has no purely imaginary eigenvalues. In practice, this means 
	checking that the real part of each eigenvalue has magnitude greater than some 
	small numerical tolerance.
\item
	Construct a matrix $\left[ \begin{smallmatrix} X_1 \\ X_2 \end{smallmatrix} \right] $ such that its 
	columns are the eigenvectors of $H$ that correspond to eigenvalues with negative 
	real part.
	Check that $X_1$ and $X_2$ are non-singular and calculate $X = X_2 X_1^{-1}$. 
	The matrix $X$ is a non-singular solution to the ARE (\ref{eqn:ric1}). 
\item
	Find the eigenvalues and eigenvectors of $X$. These will 
	occur in complex conjugate pairs. Hence, construct a 
	diagonal matrix $\Lambda$ 
	with entries being the eigenvalues of $X$ and with complex conjugate 
	eigenvalues in adjacent columns. Construct a matrix $V$ with 
	columns being the corresponding eigenvectors of $X$ normalized to have unit norm 
	and with complex conjugate eigenvectors in adjacent columns. 
\item
	Construct the $n \times n$ diagonal matrix $\tilde{\Lambda}$ with alternating 
	diagonal entries $i$ and $-i$. Also 
	construct the $n \times n$ block diagonal matrix $\tilde{V}$ with 
	each diagonal block corresponding to $\frac{1}{\sqrt{2}}
	\left[ \begin{smallmatrix}1 & 1 \\ i & -i
	\end{smallmatrix} \right]$.
\item
	Calculate $T = \tilde{V} D V^{\dagger}$ where
	$D = \left( \tilde{\Lambda}^{-1} \Lambda \right)^{\frac{1}{2}}$.
\item
	Construct
	\begin{eqnarray*}
		A &=& T \tilde{A} T^{-1};\\
		B_u &=& T \tilde{B_u};\\
		C &=& \tilde{C} T^{-1};\\
		B_{v_1} &=& \Theta C^T \diag (J).
	\end{eqnarray*}
\end{enumerate}
Then $$G(s) = C(sI - A)^{-1}B_u$$ 
and the system (\ref{eqn:model2}) corresponding to 
$\left\{ \tilde{A},\tilde{B_u},\tilde{C} \right\}$ 
is physically realizable by 
introducing only direct feedthrough quantum noises
with $B_{v_1}$ constructed as above.
No additional quantum noises are necessary for physical realizability. 

\begin{remark}
	We now justify the above numerical algorithm for constructing 
	$T$ such that $X = T^T \Theta T$. 
	
	Since $X$ is skew-symmetric, $X = V \Lambda V^\dagger,$ 
	where $V$ is a unitary matrix with columns which are the eigenvectors of $X$. 
	Also, $\Lambda$ is a diagonal matrix where the diagonal elements are the 
	eigenvalues of $X$, which are purely imaginary and occur in complex 
	conjugate pairs. For every eigenvector $v$ of $X$ corresponding to 
	eigenvalue $\lambda$, its complex conjugate $\bar{v}$ is also an eigenvector 
	and corresponds to $\bar{\lambda}$. If necessary, we reorder the columns 
	of $V$ and corresponding entries of $\Lambda$ such that these complex 
	conjugate pairs are adjacent: $V = \left[ \begin{smallmatrix}
		v_1 & \bar{v}_1 & v_2 & \bar{v}_2 & \ldots 
	\end{smallmatrix} \right].$
	
	Similarly we can write $\Theta = 
	\tilde{V}\tilde{\Lambda} \tilde{V}^\dagger$ where $\tilde{V}$ is a block 
	diagonal matrix with repeated blocks $\frac{1}{\sqrt{2}}
	\left[ \begin{smallmatrix}1 & 1 \\ i & -i
	\end{smallmatrix} \right]$ and $\tilde{\Lambda}$ is a 
	diagonal matrix with alternating entries $i$ and $-i$.
	Observe that there exists a diagonal matrix 
	$$D = \begin{bmatrix} 
		d_1 & 0 & 0 & 0 & \cdots \\
		0 & d_1 & 0 & 0 & \cdots \\
		0 & 0 & d_2 & 0 & \cdots \\
		0 & 0 & 0 & d_2 & \cdots \\
		\vdots & & & & \ddots
	\end{bmatrix}$$
	such that $\Lambda = 
	D \tilde{\Lambda} D$ and the diagonal elements $d_i$ are real and positive. 
	We now have $X = VD\tilde{V}^\dagger \Theta
	\tilde{V}DV^\dagger.$ Define $T = \tilde{V}DV^\dagger$, then 
	$X = T^\dagger \Theta T$.
	
	Observe that the matrices 
	$\left[ \begin{smallmatrix}1 & 1 \\ i & -i
	\end{smallmatrix} \right]$ and
	$\left[ \begin{smallmatrix}d_i & 0 \\ 0 & d_i
	\end{smallmatrix} \right]$ commute. Therefore $\tilde{V}$ and $D$ 
	commute. We claim that the matrix $T =D\tilde{V}V^\dagger$ is real. 
	This follows since 
	\begin{align*}
		(\tilde{V}V^\dagger)^\dagger 
	&= 
		V \tilde{V}^\dagger \linebreak \\
	&= 
		\left[ \begin{smallmatrix}
			v_1 & \bar{v}_1 & v_2 & \bar{v}_2 & \ldots 
		\end{smallmatrix} \right] \mbox{diag}  
		\left( \begin{smallmatrix}\frac{1}{\sqrt{2}}\end{smallmatrix} 
			\left[ \begin{smallmatrix} 1 & -i \\ 1 & i
			\end{smallmatrix} \right] \right) \\
	&=	
		\begin{smallmatrix}\frac{1}{\sqrt{2}}\end{smallmatrix} 
		\left[ \begin{smallmatrix}
			(v_1 + \bar{v}_1) \: &
			-i(v_1 - \bar{v}_1) \: &
			(v_2 + \bar{v}_2) \: &
			-i(v_2 - \bar{v}_2) \:
			\ldots &
		\end{smallmatrix} \right]
	\end{align*}
	which is real. Therefore $T$ as constructed above is real and 
	$X = T^T\Theta T.$
\end{remark}

\section{A suboptimal coherent Quantum LQG controller design algorithm} \label{sec:lqg}
In this section, we 
use the results in Section \ref{sec:result} to provide an algorithm for 
the design of a suboptimal coherent quantum LQG controller.
The main idea of our algorithm is to design a classical LQG 
controller and then use Theorem 2 to implement this controller as a 
quantum system. In contrast with 
the classical LQG controller synthesis problem, here 
the separation principle of using the optimal state 
estimator and optimal regulator no longer applies due to the relation between 
the optimal regulator gain and the additional quantum noises that arise when 
implementing the controller as a quantum system. This is addressed in the 
algorithm that follows.

\subsection{Problem Formulation} \label{subsec:formulation}
We will now formally state the problem to be addressed. Our formulation 
follows that in \cite{NJP1} with some minor differences. 

Suppose we have a quantum plant described by the following QSDEs which are a special case 
of~(\ref{eqn:model}):
\begin{eqnarray}
	\dif x(t) &=& A x(t) \dif t + B_u \dif u(t) + B_{w_1} \dif w_1(t) 
		; \nonumber \\
	\dif y(t) &=& C x(t) \dif t + D_u \dif u(t) + 
	D_{w_1} \dif w_1(t).
	\label{eqn:plant}
\end{eqnarray}

Here, as in (\ref{eqn:model}), $x$ 
is a column vector of $n$ self-adjoint system variables.
The column vector $\dif u(t)$ represents the input to the system. 
It consists of $n_u$ signals of the form 
$\dif u(t) = \beta_u(t) \dif t + \dif \tilde{u}(t)$ where $\dif \tilde{u}(t)$
is the noise part of $\dif u(t)$ with Ito products
$\dif \tilde{u}(t) \dif \tilde{u}^T(t) = F_{\tilde{u}} \dif t$
where $F_{\tilde{u}}$ is non-negative Hermitian. 
Also, the self-adjoint, adapted process $\beta_u(t)$ is the signal part of $\dif u(t)$.  
Furthermore, $\dif w_1$ is a column vector of $n_{w_1}$ non-commutative quantum Wiener processes with Ito 
products $\dif w_1(t) \dif {w_1}^T(t) = F_{w_1} \dif t$
where $F_{w_1}$ is non-negative Hermitian. Here, $\dif w_1$ represents noises driving 
the system and may for example include vacuum noises and/or thermal noises.
The column vector of $n_y$ signals $\dif y(t)$ represents the output of the 
system. 
Finally, as in~(\ref{eqn:model}), 
$n$, $n_u$, $n_{w_1}$ and $n_y$ are all assumed to 
be even and
$A$, $B_u$, $B_{w_1}$, $C$, $D_u$ and $D_{w_1}$ 
are appropriately 
dimensioned real matrices describing the dynamics of the system. For further
details see \cite{JNP1,NJP1}.
For simplicity we restrict our attention to the case where $n_y = n_u$.

Furthermore, suppose that we wish to minimize an infinite 
	horizon quadratic cost function:
	\begin{equation}
		\label{eqn:cost}
		J = \limsup_{t_f \to \infty} \frac{1}{t_f} 
		\int^{t_f}_0 \left< 
		x^T\!(t) R_1 x(t) 
		+ \beta_u(t)^T R_2 \beta_u(t) \right> \dif t   
	\end{equation}
	where $\left<.\right>$ denotes the quantum expectation; e.g., see \cite{NJP1}.

We restrict attention to controllers described by the following 
QSDEs which are of the form (\ref{eqn:model2}):

\begin{eqnarray}
	\dif x_K(t) &=& A_K x_K(t) \dif t + B_y \dif y(t) \nonumber \\
	&& {} + B_{v_1} \dif v_1(t) +  B_{v_2} \dif v_2(t)  
	; \nonumber \\
	\dif u(t) &=& C_K x_K(t) \dif t + \dif v_1(t).
	\label{eqn:controller}
\end{eqnarray}

The problem is to design a physically realizable quantum 
controller of the form (\ref{eqn:controller}) that minimizes the cost function
 (\ref{eqn:cost}).

To obtain an explicit expression for $J$, we consider the closed loop system:
	$$ \dif \eta = \mathcal{A} \eta \dif t + \mathcal{B} w_{CL};$$
where,
	$$
	\eta = 
	\begin{bmatrix}
		x \\
		x_K
	\end{bmatrix}; \qquad
	w_{CL} =  
	\begin{bmatrix}
		\dif w_1 \\ \dif v_1 \\ \dif v_2
	\end{bmatrix};$$
	$$\dif w_{CL} \dif {w_{CL}}^T =  F_{w_{CL}} \dif t; \qquad
	S_{w_{CL}} = \mathfrak{Re} ( F_{w_{CL}} );
	$$
	\begin{eqnarray*}
	\mathcal{A} &=&   
	\begin{bmatrix}
		A & B_u C_K \\ B_y C & A_K + B_y D_u C_K
	\end{bmatrix}; \\
	\mathcal{B} &=&  
	\begin{bmatrix}
		B_{w_1} & B_u & 0 \\ 
		B_y D_{w_1} & B_y D_u + B_{v_1} &  B_{v_2}
	\end{bmatrix}.
	\end{eqnarray*}

Finally, 
\begin{equation}
J = \mbox{Tr} \left( \bar{R} \bar{Q} \right),
	\label{eqn:J}
\end{equation}
where 
$\bar{Q}$ is the unique symmetric positive definite solution of the 
Lyapunov equation
\begin{equation*}
	\mathcal{A} \bar{Q} + \bar{Q} \mathcal{A}^T 
	+ \mathcal{B} S_{w_{CL}} \mathcal{B}^T = 0;
\end{equation*}
	and
$$\bar{R} = \begin{bmatrix} R_1 & 0 \\ 0 & C_K^T R_2 C_K \end{bmatrix}. $$
For a more detailed derivation of these expressions, see~\cite{NJP1}.

In the following subsection, we present an algorithm for designing a quantum controller 
of the form (\ref{eqn:controller}) which is a suboptimal solution to this problem.

\subsection{Design Algorithm} \label{subsec:algorithm}
We start by forming an \emph{Auxiliary Classical LQG Problem}. Consider the plant 
and controller equations (\ref{eqn:plant}) and (\ref{eqn:controller}) and define 
$\dif \hat{u} = \dif u - \dif v_1$. By temporarily ignoring 
the $B_{v_1}\dif v_1$ and $B_{v_2}\dif v_2$ noise terms, and treating $\dif w_1$ and $\dif v_1$ as 
classical Wiener processes with intensity matrices $S_{w_1}$ and $S_{v_1}$
respectively, 
we obtain the auxiliary classical plant equations: 
\begin{eqnarray}
	\dif x &=& A x \dif t + B_u \dif \hat{u} 
	+ B_{w_1} \dif w_1 + B_u \dif v_1; \nonumber \\
	\dif y &=& C x \dif t + D_u \dif \hat{u}  
	+ D_{w_1} \dif w_1 + D_u \dif v_1;
	\label{eqn:auxplant}
\end{eqnarray}
and the auxiliary classical controller equations:
\begin{eqnarray}
	\dif x_K &=& A_K x_K \dif t +
	B_{y} \dif y; \nonumber \\
	\dif \hat{u} &=& C_K x_K \dif t. 
	\label{eqn:auxcontroller}
\end{eqnarray}

We also define an \emph{Auxiliary Cost Function} which introduces an extra term 
to account for the fact that we have ignored the noise terms 
$B_{v_1} \dif v_1$ and $B_{v_2} \dif v_2$ that will appear in the quantum version of the controller: 
\begin{align}
	J_{AUX} &= 
	\limsup_{t_f \to \infty} \frac{1}{t_f} 
	\mathbb{E} \! \left[
	\int^{t_f}_0 \!\!\! 
		x^T\!(t) R_1 x(t) 
		+ \beta_u(t)^T R_2 \beta_u(t) \dif t 
	\right]
	\nonumber \\
	&\qquad + \limsup_{t_f \to \infty} \frac{1}{t_f} 
	\mathbb{E} \! \left[
	\int^{t_f}_0 \!\!
	\rho \beta_u(t)^T R_2 \beta_u(t) 
	\dif t \right] \label{eqn:Jaux}
\end{align}
where $\mathbb{E}\!\left[.\right]$ denotes the classical expectation and $\rho \ge 0$. 
Equivalently,
\begin{equation}
	J_{AUX} = \limsup_{t_f \to \infty} \frac{1}{t_f} 
	\mathbb{E} \! \left[ \int^{t_f}_0 \!\!\!  
	x^T\!(t) R_1 x(t)	+ \beta_u(t)^T \tilde{R_2} \beta_u(t) 
	\! \dif t \right] 
	\label{eqn:Jaux2}
\end{equation}
where $\tilde{R_2} = (1 + \rho) R_2$, $\rho \ge 0$. 
The \emph{Auxiliary LQG problem} is to find a 
	controller (\ref{eqn:auxcontroller}) that minimizes the 
	cost function (\ref{eqn:Jaux2}) for the 
	plant (\ref{eqn:auxplant}).

Our approach to the coherent quantum LQG problem is as follows. 
The auxiliary LQG problem is first solved for a given 
$\rho \ge 0$ and the resulting auxiliary controller (\ref{eqn:auxcontroller}) 
is implemented as a quantum controller (\ref{eqn:controller}) 
by applying Theorem \ref{thm:tf} or Theorem \ref{thm:main}. The 
cost function (\ref{eqn:cost}) is then evaluated using the expression (\ref{eqn:J}). 
Finally, this process is repeated, optimizing the cost function (\ref{eqn:cost}) 
by using a line search over the parameter $\rho$ to obtain our final suboptimal 
controller.

We now detail one iteration of this design process. 
The auxiliary LQG problem is a standard classical LQG problem and is 
solved in the usual manner; see for example~\cite{KS72}. The solution is the 
auxiliary controller (\ref{eqn:auxcontroller}) with
\begin{eqnarray*}
	A_K &=& A - KC - B_uF + KD_uF; \\
	B_{y} &=& K; \\
	C_K &=& -F.
\end{eqnarray*}
Here, $F$ and $K$ are obtained as follows:
$$F = \tilde{R}_2^{-1} B_u^T P;$$
where $P \ge 0$ is the solution to the ARE 
$$R_1 - PB_u \tilde{R}_2^{-1}B_u^TP + A^TP + PA = 0,$$
and 
$$K = (Q C^T + V_{12} ) V_2^{-1};$$
where $Q \ge 0$ is the solution to the ARE 
\begin{eqnarray*}
0 &=&  (A - V_{12} V_2^{-1}C) Q 
 + Q (A - V_{12} V_2^{-1}C)^T \\
&& {} - Q C^T V_2^{-1} C Q  
+ V_1 - V_{12} V_2^{-1}V_{12}^T.
\end{eqnarray*}
Here,
\begin{align*}
& \mathbb{E} 
\begin{bmatrix} B_{w_1} & B_u \\ D_{w_1} & D_u \end{bmatrix} 
\begin{bmatrix} \dif w_1 \\ \dif v_1 \end{bmatrix}
\begin{bmatrix} \dif w_1 \\ \dif v_1 \end{bmatrix}^T
\begin{bmatrix} B_{w_1} & B_u \\ D_{w_1} & D_u \end{bmatrix}^T\\
& \qquad =
\begin{bmatrix} V_1 & V_{12} \\ V_{12}^T & V_2 \end{bmatrix} \dif t;\\
& \begin{bmatrix} V_1 & V_{12} \\ V_{12}^T & V_2 \end{bmatrix} \\
& \qquad =
\begin{bmatrix} B_{w_1} & B_u \\ D_{w_1} & D_u \end{bmatrix}
\begin{bmatrix} S_{w_1} & 0 \\ 0 & S_{v_1} \end{bmatrix}
\begin{bmatrix} B_{w_1} & B_u \\ D_{w_1} & D_u \end{bmatrix}^T.
\end{align*}

Next, we obtain a fully quantum system of the form (\ref{eqn:controller}), based on 
the auxiliary controller (\ref{eqn:auxcontroller}) with 
$\left\{ A_K, B_{y}, C_K \right\}$ obtained above. We first attempt to apply 
Theorem~\ref{thm:tf}. If the conditions of the theorem are satisfied, the transfer 
function of the auxiliary controller is implemented as a system 
(\ref{eqn:controller}), with only direct feedthrough quantum noises 
introduced by applying Theorem 
\ref{thm:tf}. That is, $n_v = n_u$.  If the conditions of Theorem~\ref{thm:tf} 
are not satisfied, then the auxiliary controller is implemented 
by applying Theorem~\ref{thm:main}, which will result in $n_v > n_u$ quantum noises.

Finally, the cost function (\ref{eqn:cost}) is evaluated using the expression (\ref{eqn:J}). 
For details on obtaining $B_{v_1}$ and $B_{v_2}$ see Section~\ref{subsec:gc}. 

Our algorithm is summarized as follows:

\begin{enumerate}
	\item For a given $\rho \ge 0$, form the Auxiliary Classical LQG Problem 
		(\ref{eqn:auxplant}), (\ref{eqn:Jaux}). 
	\item Solve to obtain the classical auxiliary controller (\ref{eqn:auxcontroller}).
	\item Implement this controller as a coherent quantum controller (\ref{eqn:controller}).
	\item Form the corresponding closed loop system, and 
		evaluate the resulting cost function (\ref{eqn:cost}).
	\item Repeat, optimizing over $\rho \ge 0$.
\end{enumerate}

We now give a heuristic motivation for our algorithm.
In the standard (classical) LQG problem, the separation principle allows the 
optimal state estimator and optimal regulator to be designed independently and 
then combined to yield the optimal controller. In 
contrast to this, in the quantum version of the problem, the 
regulator gain $C_K$ directly affects how strongly the quantum noises $\dif v_1$ and  
$\dif v_2$ impact the state estimator because $B_{v_1}$ and $B_{v_2}$
depend on $C_K$. 

Our method ignores the introduction of the additional noises 
$\dif v_1$ and $\dif v_2$ when designing our state estimator. In order to ensure 
that the effect of these noises 
is not too great, when designing the regulator 
we introduce the parameter $\rho$ which puts an additional penalty 
on the size of the control signal. 
The final step of optimizing over $\rho$ ensures the right balance: if $\rho$ 
is too small the effect of the additional noises $\dif v_1$ and $\dif v_2$ dominate 
the closed loop system response leading to poor performance 
whereas if $\rho$ is too large, the 
feedback gain is unduly penalized also leading to poor performance.

We justify our approach by observing that in practice, our algorithm is 
computationally tractable and examples show that the controllers so obtained 
yield good results. In particular, the 
example which follows demonstrates how a suboptimal coherent quantum controller can 
outperform a combination of heterodyne measurement and optimal (classical) measurement 
based feedback control.

\section{Illustrative Example}\label{sec:ex}
To demonstrate the coherent quantum controller 
design process of Section \ref{sec:lqg}, we consider 
a two mirror optical cavity driven by thermal noise of intensity 
$k_n$ as shown in Figure \ref{fig:ex}. This example is a modification of an example considered in \cite{HM13}. 
Optically coupled to the second mirror is a controller to be designed via our 
algorithm. We compare the value of the \emph{cost function} obtained with our controller 
to the no control case; i.e. when the second mirror is not connected to any other 
system and thus driven by a vacuum noise. We also compare our approach with a  
scheme involving heterodyne detection and an optimal classical 
LQG controller. Our design objective is to minimize the expected number of photons in 
the cavity: $N = \left< a^{\dagger}a \right> = \frac{1}{4} 
\left< {x_1}^2 + {x_2}^2 \right> - 0.5$, where 
$x_1$ and $x_2$ are the position and momentum operators for the cavity.
It will be shown 
that for all $k_n > 0$,
it is possible to achieve better performance using a 
quantum controller designed using our method than with the optimal 
measurement based controller. This validates the utility of our method.

\begin{figure}[h]
	\includegraphics[trim = 0mm 0mm 0mm -10mm, scale=0.4]{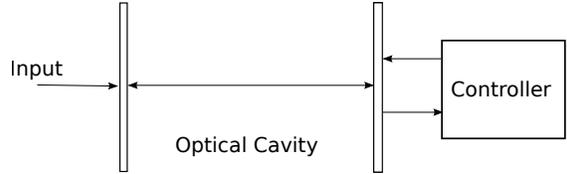}
	\caption{\label{fig:ex} Quantum Plant: we wish to minimize the expected number of photons 
in the cavity.}
\end{figure}
Our plant is of the form (\ref{eqn:plant}) with 
\begin{eqnarray}
	A &=&  {} - \frac{\gamma}{2} I_{2 \times 2}; \nonumber \\
	B_u &=&  - \sqrt{\kappa_2} I_{2 \times 2}; \nonumber \\
	B_{w_1} &=& \sqrt{\kappa_1} I_{2 \times 2}; \nonumber \\
	C &=& \sqrt{\kappa_2} I_{2 \times 2}; \nonumber \\
	D_u &=& I_{2 \times 2}; \nonumber \\
	D_{w_1} &=& 0_{2 \times 2}; \nonumber \\
	S_{w_1} &=& (1 + 2 k_n) I_{2 \times 2}\nonumber \\    
	\gamma &=&  0.2; \nonumber \\
	\kappa_1 &=&  0.1; \nonumber \\
	\kappa_2 &=&  0.1.
	\label{eqn:explant}
\end{eqnarray}

It is sufficient to minimize the cost function (\ref{eqn:cost}) with 
$R_1 = I_{2 \times 2}$ and $R_2 = 0_{2 \times 2}$. Then $N = \frac{1}{4}J - 0.5$.

\begin{remark}
We use $R_2 = 10^{-6} I_{2 \times 2}$ to design 
both the heterodyne classical LQG controller and the auxiliary LQG controller. However, 
the cost function (\ref{eqn:cost}) for the resulting controller is then 
evaluated using $R_2 = 0_{2 \times 2}$.
\end{remark}

\subsection{No Control}
We first consider the case of no control as a reference. Here, the 
system input is driven by vacuum noise: $\dif u = \dif v_1$.
The resulting closed loop system is:
$$\dif x = A x \dif t + \begin{bmatrix} B_{w_1} & B_u \end{bmatrix}
\begin{bmatrix} \dif w_1 \\ \dif v_1 \end{bmatrix},$$
where 
$
\dif w_{CL} = \left[ \begin{smallmatrix}
		\dif w_1 \\ \dif v_1
	\end{smallmatrix} \right]
	$ is a quantum Wiener process with covariance
	$$S_{w_{CL}} = \begin{bmatrix}
	(1 + 2k_n) I_{2 \times 2} & 0 \\
	0 & I_{2 \times 2} 
\end{bmatrix}.$$
Equation (\ref{eqn:J}) can then be used to find 
the value of the cost function (\ref{eqn:cost}) for a range of values for $k_n$.

\subsection{Heterodyne measurement and classical LQG control}
Next we consider combining heterodyne measurement with a classical optimal LQG 
controller. 
Heterodyne measurement introduces an additional vacuum noise 
input. Similarly, the output of the classical controller will 
contain a vacuum noise component when applied to the input of the plant. This is accounted for 
by making the following substitutions: 
\begin{eqnarray*}
	\dif u &=&  \dif \tilde{u} + \dif w_3; \\
	\dif \tilde{y} &=& \dif y + \dif w_2;
\end{eqnarray*}
into (\ref{eqn:plant}) to obtain an \emph{augmented plant}. 
Here $\dif \tilde{u}$ and $\dif \tilde{y}$ are classical signals which represent 
the input and output of the augmented plant. 
The resulting equations for the augmented plant are as follows:
\begin{eqnarray*}
	\dif x &=&  A x \dif t + B_u \dif \tilde{u}  
	+ B_{\tilde{w}} \dif \tilde{w} ; \nonumber \\
	\dif \tilde{y} &=&  C x \dif t + D_u \dif \tilde{u} 
	+ D_{\tilde{w}} \dif \tilde{w} .
\end{eqnarray*}
Here $A$, $B_u$, $C$, and $D_u$ are as before, and
\begin{eqnarray*}
	B_{\tilde{w}} &=& \begin{bmatrix}
		{} - \sqrt{\kappa_1} I_{2 \times 2} &
		0_{2 \times 2} &
		{} - \sqrt{\kappa_2} I_{2 \times 2}
	\end{bmatrix}; \\
	D_{\tilde{w}} &=& \begin{bmatrix}
		0_{2 \times 2} &
		I_{2 \times 2} &
		I_{2 \times 2}
	\end{bmatrix}; \\
	\dif \tilde{w} &=& \begin{bmatrix}
		\dif w_1 \\ \dif w_2 \\ \dif w_3
	\end{bmatrix}; \\
	S_{\tilde{w}} &=&  
	\begin{bmatrix}
	(1 + 2k_n) I_{2 \times 2} & 0 & 0 \\
	0 & I_{2 \times 2} & 0 \\ 
	0 & 0 & I_{2 \times 2}  
\end{bmatrix}.
\end{eqnarray*}

As with the auxiliary LQG problem in our design algorithm, we treat 
$\dif \tilde{w}$ as a standard classical Wiener process with intensity matrix 
$S_{\tilde{w}}$. We now have a standard classical LQG problem. We wish to find a 
controller of the form:
\begin{eqnarray*}
	\dif x_K &=& A_K x_K \dif t +
	B_{y} \dif \tilde{y}; \nonumber \\
	\dif \tilde{u} &=& C_K x_K \dif t. 
\end{eqnarray*}

The estimator gain $K$ and the regulator gain $F$ are obtained in the usual 
manner. For this example $R_2 = 0_{2 \times 2}$ but for computational reasons we 
assume $R_2$ takes a small value of $R_2 = 10^{-6}$ 
and hence $\left\{ A_K, B_y, C_K \right\}$ are obtained. The closed loop 
system is then as follows:
\begin{eqnarray}
	\begin{bmatrix}
		\dif x \\
		\dif x_K
	\end{bmatrix}
	&=& 
	\mathcal{A}
	\begin{bmatrix}
		x \\
		x_K
	\end{bmatrix}
	\dif t +
	\mathcal{B}
	\dif \tilde{w};	\nonumber \\
	\mathcal{A} &=&  
	\begin{bmatrix}
		A & B_u C_K \\ B_y C & A_K + B_y D_u C_K
	\end{bmatrix}; \nonumber \\
	\mathcal{B} &=& 
	\begin{bmatrix}
		B_{\tilde{w}} \\ B_y D_{\tilde{w}}
	\end{bmatrix}. \nonumber
\end{eqnarray}

The value of the cost function (\ref{eqn:cost}) can now be calculated using (\ref{eqn:J}) 
with $R_2 = 0_{2 \times 2}$ and
substituting $S_{\tilde{w}}$ for $S_{w_{CL}}$. 

\subsection{Quantum LQG control}
Finally we consider our proposed control scheme. 
First the auxiliary LQG problem is formed. The auxiliary plant is
given by (\ref{eqn:auxplant}) with parameters as in (\ref{eqn:explant}). 
The cost function is given by (\ref{eqn:Jaux}) with 
$R_1 = I_{2 \times 2}$ and $R_2 = 10^{-6}$.

Then, optimizing $J$ over $\rho$, we do the following: 
\begin{enumerate}
	\item Solve the auxiliary LQG problem as detailed above  
		to obtain $\left\{ A_K, B_y, C_K \right\}$. 
	\item Obtain a physically realizable quantum implementation of 
		$\left\{ A_K, B_y, C_K \right\}$. 
		We do this by first attempting to apply Theorem~3 to obtain $B_{K1}$. 
		If the conditions of 
		this theorem are not met we apply Theorem 2 to obtain $B_{K1}$ and $B_{K2}$.
	\item Evaluate the cost function (\ref{eqn:cost}) using (\ref{eqn:J}) 
		with $R_1 =  I_{2 \times 2}$ and $R_2 = 0_{2 \times 2}$.
\end{enumerate}

\subsection{Comparison of controller performance}
The relative performance of the no control case, the classical LQG case and our coherent control case 
are illustrated in Figure \ref{fig:results1} and Figure \ref{fig:results2}. 
In the regime where both the thermal
noise driving the system and the quantum noises are significant, 
the coherent quantum feedback controller offers the best performance of the 
schemes considered. 
If we leave the quantum regime with $k_n \gg 0$, the
relative performance benefits of the coherent quantum feedback controller 
over the measurement based feedback controller diminish
as the thermal noise dominates the system and the quantum noises become
insignificant by comparison.
In the limit as $k_n \to 0$, where the system is driven only by vacuum noise, 
our proposed controller offers no advantage over the no control case. 
This is consistent with the idea that the cavity cannot be driven below the
vacuum state.

\begin{figure}[h]
	\centering
	\includegraphics[trim = 0mm 0mm 0mm 0mm, scale=0.5]{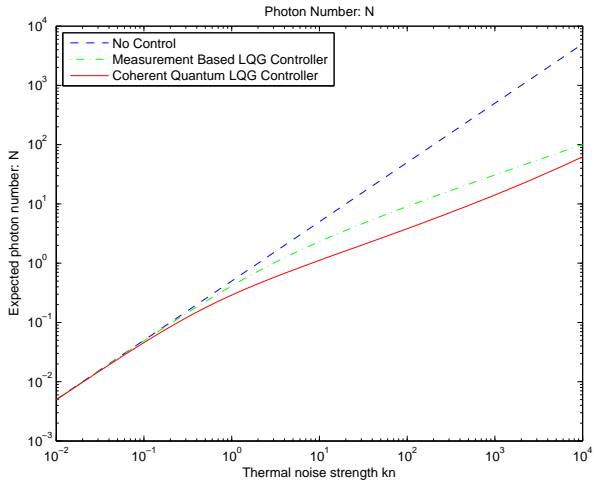}
	\caption{\label{fig:results1} Expected photon number $N$ for no control,  
	heterodyne measurement and classical LQG control, and 
	coherent quantum LQG control.}
\end{figure}

\begin{figure}[h]
	\centering
\includegraphics[trim = 0mm 0mm 0mm 0mm, scale=0.5]{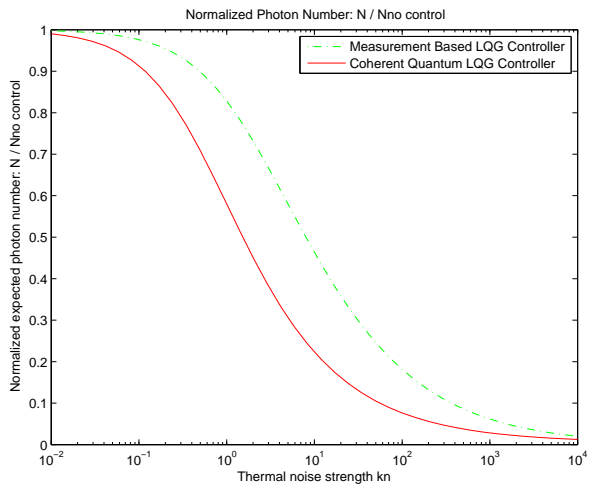}
\caption{\label{fig:results2} Normalized photon number $N / N_{NC}$ for 
	heterodyne measurement and classical LQG control, and 
	coherent quantum LQG control.}
\end{figure}

\section{Conclusion} \label{sec:conc}
The notion of physical realizability is fundamental to the coherent quantum feedback control problem 
where we wish to implement a given synthesized controller as a quantum system. 
By introducing additional quantum noises, it is always 
possible to make a given LTI system physically realizable. However, 
introducing additional quantum noises is undesirable in terms of the control system performance. 
In this paper, we have 
given an expression for the number of introduced quantum noises that are 
necessary to make a given LTI system physically realizable. 
Our result also gives 
a method for constructing the resulting fully quantum system.

We also considered the case where a strictly proper transfer function 
is to be physically realized. 
We have given a result 
in terms of a Riccati equation for when it is possible to 
physically realize a specified transfer function by only introducing 
direct feedthrough vacuum noises and no additional quantum noises. 
We have also given conditions for when this Riccati equation has a suitable solution. 

Using these results we have developed
an algorithm for obtaining a suboptimal solution to a coherent quantum LQG 
control problem. 
Our example demonstrates the utility of our results and shows that
coherent quantum feedback control can offer performance benefits over
measurement based feedback control.

%
%

\bibliography{../irpnew}  

\begin{thebibliography}{10}
\providecommand{\url}[1]{#1}
\csname url@samestyle\endcsname
\providecommand{\newblock}{\relax}
\providecommand{\bibinfo}[2]{#2}
\providecommand{\BIBentrySTDinterwordspacing}{\spaceskip=0pt\relax}
\providecommand{\BIBentryALTinterwordstretchfactor}{4}
\providecommand{\BIBentryALTinterwordspacing}{\spaceskip=\fontdimen2\font plus
\BIBentryALTinterwordstretchfactor\fontdimen3\font minus
  \fontdimen4\font\relax}
\providecommand{\BIBforeignlanguage}[2]{{%
\expandafter\ifx\csname l@#1\endcsname\relax
\typeout{** WARNING: IEEEtran.bst: No hyphenation pattern has been}%
\typeout{** loaded for the language `#1'. Using the pattern for}%
\typeout{** the default language instead.}%
\else
\language=\csname l@#1\endcsname
\fi
#2}}
\providecommand{\BIBdecl}{\relax}
\BIBdecl

\bibitem{GZ00}
C.~Gardiner and P.~Zoller, \emph{Quantum Noise}.\hskip 1em plus 0.5em minus
  0.4em\relax Berlin: Springer, 2000.

\bibitem{AFP09}
G.~Auletta, M.~Fortunato, and G.~Parisi, \emph{{Quantum Mechanics}}, ser.
  Quantum Mechanics.\hskip 1em plus 0.5em minus 0.4em\relax Cambridge
  University Press, 2009.

\bibitem{JNP1}
M.~R. James, H.~I. Nurdin, and I.~R. Petersen, ``${H}^\infty$ control of linear
  quantum stochastic systems,'' \emph{IEEE Transactions on Automatic Control},
  vol.~53, no.~8, pp. 1787--1803, 2008.

\bibitem{ShP12}
A.~J. Shaiju and I.~R. Petersen, ``A frequency domain condition for the
  physical realizability of linear quantum systems,'' \emph{IEEE Transactions
  on Automatic Control}, vol.~57, no.~8, pp. 2033--2044, 2012.

\bibitem{NJP1}
H.~I. Nurdin, M.~R. James, and I.~R. Petersen, ``Coherent quantum {LQG}
  control,'' \emph{Automatica}, vol.~45, no.~8, pp. 1837--1846, 2009.

\bibitem{WM10}
H.~M. Wiseman and G.~J. Milburn, \emph{Quantum Measurement and Control}.\hskip
  1em plus 0.5em minus 0.4em\relax Cambridge University Press, 2010.

\bibitem{VMS12}
R.~Vijay, C.~Macklin, D.~H. Slichter, S.~J. Weber, K.~W. Murch, R.~Naik, A.~N.
  Korotkov, and I.~Siddiqi, ``Stabilizing {R}abi oscillations in a
  superconducting qubit using quantum feedback,'' \emph{Nature}, vol. 490, pp.
  77--80, 2012.

\bibitem{MaP4}
A.~I. Maalouf and I.~R. Petersen, ``Coherent ${H}^{\infty}$ control for a class
  of linear complex quantum systems,'' \emph{IEEE Transactions on Automatic
  Control}, vol.~56, no.~2, pp. 309--319, 2011.

\bibitem{MAB08}
H.~Mabuchi, ``Coherent-feedback quantum control with a dynamic compensator,''
  \emph{Physical Review A}, vol.~78, p. 032323, 2008.

\bibitem{MaP3}
A.~I. Maalouf and I.~R. Petersen, ``Bounded real properties for a class of
  linear complex quantum systems,'' \emph{IEEE Transactions on Automatic
  Control}, vol.~56, no.~4, pp. 786 -- 801, 2011.

\bibitem{VuP12b}
S.~L. Vuglar and I.~R. Petersen, ``Singular perturbation approximations for
  general linear quantum systems,'' in \emph{Proceedings of the Australian
  Control Conference}, Sydney, Australia, Nov 2012, pp. 459--463,
  {arXiv:1208.6155 [quant-ph]}.

\bibitem{NJD09}
H.~I. Nurdin, M.~R. James, and A.~C. Doherty, ``Network synthesis of linear
  dynamical quantum stochastic systems,'' \emph{SIAM Journal on Control and
  Optimization}, vol.~48, no.~4, pp. 2686--2718, 2009.

\bibitem{NUR10}
H.~Nurdin, ``Synthesis of linear quantum stochastic systems via quantum
  feedback networks,'' \emph{IEEE Transactions on Automatic Control}, vol.~55,
  no.~4, pp. 1008 --1013, April 2010.

\bibitem{NUR10A}
------, ``On synthesis of linear quantum stochastic systems by pure
  cascading,'' \emph{IEEE Transactions on Automatic Control}, vol.~55, no.~10,
  pp. 2439 --2444, October 2010.

\bibitem{Pet11}
I.~R. Petersen, ``Cascade cavity realization for a class of complex transfer
  functions arising in coherent quantum feedback control,'' \emph{Automatica},
  vol.~47, no.~8, pp. 1757 -- 1763, 2011.

\bibitem{VuP11a}
S.~L. Vuglar and I.~R. Petersen, ``How many quantum noises need to be added to
  make an {LTI} system physically realizable?'' in \emph{Proceedings of the
  Australian Control Conference}, Melbourne, Australia, November 2011.

\bibitem{VuP12a}
------, ``A numerical condition for the physical realizability of a quantum
  linear system,'' in \emph{Proceedings of the 20th International Symposium on
  Mathematical Theory of Networks and Systems}, Melbourne, Australia, 2012.

\bibitem{VuP12c}
------, ``Quantum implemention of an {LTI} {S}ystem with the minimal number of
  additional quantum noise inputs.'' in \emph{Proceedings of the 12th biannual
  European Control Conference}, Zurich, Switzerland, 2013, {arXiv:1304.6815
  [quant-ph]}.

\bibitem{HJ85}
R.~A. Horn and C.~R. Johnson, \emph{Matrix Analysis}.\hskip 1em plus 0.5em
  minus 0.4em\relax Cambridge, UK: Cambridge University Press, 1985.

\bibitem{BER05}
D.~S. Bernstein, \emph{Matrix Mathematics: Theory, Facts, And Formulas with
  Application to Linear Systems Theory}.\hskip 1em plus 0.5em minus 0.4em\relax
  Princeton, New Jersey: Princeton University Press, 2005.

\bibitem{BAK02}
A.~Baker, \emph{Matrix Groups: An Introduction to Lie Group Theory}.\hskip 1em
  plus 0.5em minus 0.4em\relax New York: Springer-Verlag, 2002.

\bibitem{ZDG96}
K.~Zhou, J.~Doyle, and K.~Glover, \emph{Robust and Optimal Control}.\hskip 1em
  plus 0.5em minus 0.4em\relax Upper Saddle River, NJ: Prentice-Hall, 1996.

\bibitem{KS72}
H.~Kwakernaak and R.~Sivan, \emph{Linear Optimal Control Systems}.\hskip 1em
  plus 0.5em minus 0.4em\relax Wiley, 1972.

\bibitem{HM13}
R.~Hamerly and H.~Mabuchi, ``Coherent controllers for optical-feedback cooling
  of quantum oscillators,'' \emph{Phys. Rev. A}, vol.~87, no.~1, p. 013815,
  2013.

\end{thebibliography}
\bibliographystyle{IEEEtran}

\end{document}